\documentclass[a4paper,10pt]{article}

\usepackage{diagrams}
\usepackage{natbib}
\usepackage{amsmath}
\usepackage{amsthm}
\usepackage{amssymb}
\usepackage{hyperref}
\usepackage{amsfonts}
\usepackage{amscd}
\usepackage[utf8]{inputenc}

\newtheorem{theorem}{Theorem}

\newtheorem{corollary}[theorem]{Corollary}

\newtheorem{definition}[theorem]{Definition}

\newtheorem{lemma}[theorem]{Lemma}

\newtheorem{assumption}[theorem]{Assumption}
\newtheorem{proposition}[theorem]{Proposition}
\theoremstyle{remark}
\newtheorem{example}[theorem]{Example}
\newtheorem{remark}[theorem]{Remark}

\title{Every hierarchy of beliefs is a type\thanks{I thank Ferenc Forgó, Aviad Heifetz and Zoltán Kánnai for their suggestions and remarks. Naturally, all errors are mine. This work is supported by the János Bolyai Research Scholarship of the Hungarian Academy of Sciences and by grant OTKA 72856.}}
\author{Miklós Pintér \\
Corvinus University of Budapest\thanks{Department of Mathematics, Corvinus University of Budapest, 1093 Hungary, Budapest, Fővám tér 13-15., miklos.pinter@uni-corvinus.hu}}

\begin{document}

\maketitle

\begin{abstract}
When modeling game situations of incomplete information one usually considers the players' hierarchies of beliefs, a source of all sorts of complications. \cite{Harsanyi196768}'s idea henceforth referred to as the "Harsányi program" is that hierarchies of beliefs can be replaced by "types". The types constitute the "type space". In the purely measurable framework \cite{HeifetzSamet1998} formalize the concept of type spaces and prove the existence and the uniqueness of a universal type space. \cite{Meier2001} shows that the purely measurable universal type space is complete, i.e., it is a consistent object. With the aim of adding the finishing touch to these results, we will prove in this paper that in the purely measurable framework every hierarchy of beliefs can be represented by a unique element of the complete universal type space.
\end{abstract}


\section{Introduction}

It is recommended that the models of incomplete information situations be able to handle the players' hierarchies of beliefs, e.g. player $1$'s beliefs about the parameters of the game, player $1$'s beliefs about player $2$'s beliefs about the parameters of the game, player $1$'s beliefs about player $2$'s beliefs about player $1$'s beliefs about the parameters of the game, and so on to infinity. The explicit use of hierarchies of beliefs\footnote{In this paper we use the terminology hierarchy of beliefs instead of the longer coherent hierarchy of beliefs.}, however, makes the analysis very cumbersome, hence it is desirable that those should not appear explicitly in the models.

In order to make the models of incomplete information situations more amenable to analysis, \cite{Harsanyi196768} proposes to replace the hierarchies of beliefs by types. There are two approaches to examine the connection between types and hierarchies of beliefs. The first one, when we take a type space as a primitive, in this case each player's each type in the given type space defines a hierarchy of beliefs of the given player (see e.g. \cite{BattigalliSiniscalchi1999} or the proof of Proposition \ref{atetel2} in Section \ref{proof}). 

In the second case, we take hierarchies of beliefs first, then we "construct" a type space of the considered hierarchies of beliefs (see e.g. \cite{MertensZamir1985}, \cite{BrandenburgerDekel1993}, \cite{Heifetz1993}, \cite{MertensSorinZamir1994}, \cite{Pinter2005} among others). It is an open question in this field whether the two approaches are equivalent, i.e., whether there is a type space containing every hierarchy of beliefs or equivalently whether every hierarchy of beliefs is a type.   

Harsányi's main concept is that the types can substitute for the hierarchies of beliefs, and all types can be collected into an object on which the probability measures represent the players' (subjective) beliefs. Henceforth, we call this method of modeling the "Harsányi program".

However, at least two questions come up in connection with the Harsányi program: (1) Is the concept of types itself appropriate? (2) Can every hierarchy of beliefs be a type?

Question (1) consists of two subquestions. First, can all types be collected into one object? The concept of universal type space formalizes this requirement: the universal type space in a certain category of type spaces is a type space (a) which is in the given category and (b) into which, every type space of the given category can be "embedded" in a unique way. In other words, the universal type space is the most general type space, it contains all type spaces (all types). In the purely measurable framework \cite{HeifetzSamet1998} introduces the concept of universal type space and proves that the universal type space exists and is unique.

Second, can any probability measure on the object of the collected types (states of the world) be a (subjective) belief? \cite{Brandenburger2003} introduces the notion of complete type spaces: a type space is complete, if its type functions are surjective (onto). Roughly speaking, a type space is complete, if each probability measure on the object consisting of the types of the model is assigned to a type. \cite{Meier2001} shows that the purely measurable universal type space is complete. 

We can now conclude that the answer for question (1) is affirmative, i.e., in the purely measurable framework the complete universal type space exists.

Question (2) is whether the universal type space contains every hierarchy of beliefs. Mathematically, the problem is as follows: every hierarchy of beliefs defines an inverse system of probability measure spaces; the question is the following: do these inverse systems of probability measure spaces have inverse limits? Kolmogorov Extension Theorem is about this problem, however it calls for topological concepts, e.g. for (inner compact) regular probability measures. Therefore up to now, all papers on this problem, e.g. \cite{BogeEisele1979}, \cite{MertensZamir1985}, \cite{BrandenburgerDekel1993}, \cite{Heifetz1993}, \cite{MertensSorinZamir1994}, \cite{Pinter2005} among others, use topological type spaces instead of purely measurable ones. Although these papers give a positive answer to question (2), i.e., their type spaces contain all "considered" hierarchies of beliefs, very recently \cite{Pinter2010b} shows that there is no universal topological type space, i.e., there is no topological type space which contains every topological type space, therefore the answer for question (1) is negative in this case, put it differently, in the topological framework the Harsányi program fails.

\begin{table}[h!]
\begin{equation*}
\begin{array}{l|c|c}
& \text{Purely measurable setting} & \text{Topological setting} \\
\hline
\text{Question (1)} & \surd & \emptyset \\
\text{Question (2)} & \surd & \surd
\end{array}
\end{equation*}
\caption{The Harsányi program}
\label{abra0}
\end{table}

In the above mentioned papers the authors answer question (2) (affirmatively) by constructing an object consisting of all considered hierarchies of beliefs, called beliefs space, and show that the constructed beliefs space defines (is equivalent to) a topological type space.

It is worth mentioning that the above papers use different concepts of hierarchies of beliefs; e.g. in \cite{MertensZamir1985} the parameter space is a compact topological space, the beliefs are (inner closed) regular probability measures, and the events (for higher order beliefs) are the Borel sets of weak* topologies, in \cite{BrandenburgerDekel1993} the parameter space is a Polish space, the beliefs are probability measures and the events (for higher order beliefs) are the Borel sets of weak* topologies, while in \cite{Heifetz1993} the parameter space is a Hausdorff topological space, the beliefs are (inner compact) regular probability measures, and the events (for higher order beliefs) are the Borel sets of weak* topologies, and so on. Therefore, the concept of hierarchy of beliefs is different from paper to paper, from setting to setting.

In this paper we work with the category of type spaces introduced by \cite{HeifetzSamet1998}, i.e., we consider the purely measurable framework (with purely measurable type spaces, see Definition \ref{tipuster}, and with purely measurable hierarchies of beliefs, see Definition \ref{velemenyter}). Our main result is that in the purely measurable framework every hierarchy of beliefs is a type, put it differently, the Harsányi program works in the purely measurable setting.

The strategy of the proof is the same as in the above mentioned papers on topological type spaces, i.e., we construct an object such that (1) it contains every hierarchy of beliefs (see Definition \ref{velemenyter}) and (2) it generates a type space. More exactly, it is showed that the purely measurable beliefs space is equivalent to the purely measurable complete universal type space, we mean those are measurable isomorphic.

As we have already mentioned, the proof that the universal type space contains every hierarchy of beliefs is based on the Kolmogorov Extension Theorem. Since we work in the purely measurable framework, we avoid using topological concepts, and apply a non-topological variant of the Kolmogorov Extension Theorem. Mathematically speaking, we take a new result of \cite{Pinter2010a} to show that the inverse systems of probability measure spaces under consideration (the purely measurable hierarchies of beliefs) have inverse limits.

The intuition behind our result is that the purely measurable hierarchies of beliefs are special stochastic processes. No new information enters the process, i.e., the players do not learn anything new (about the states of the world) when they are ``thinking`` (considering their hierarchies of beliefs). In general, the hierarchies of beliefs, however, do not have this property, e.g. in \cite{HeifetzSamet1999}'s example the players can learn new things (about the states of the world) when they are ``thinking". 

In our opinion, the purely measurable approach, where the players do not learn anything new about the states of the world, very accurately reflects the intuition that the hierarchies of beliefs are ``only`` descriptions of the players' beliefs. Moreover, this accuracy makes it possible to achieve our positive result.  

One important remark: our result does not contradict \cite{HeifetzSamet1999}'s counterexample, because their non-type hierarchy of beliefs is not in the purely measurable beliefs space (for details see Section \ref{comparison}).

The paper is organized as follows: 
Section \ref{typespace} presents the technical setup and some basic results of the field. Our main result (Theorem \ref{fotetel}) comes on stage in Section \ref{bs}, and in Section \ref{proof} we present the proof of Theorem \ref{fotetel}. Section \ref{comparison} is for the detailed discussion of the connection between our result and two other papers \cite{HeifetzSamet1999} and \cite{Pinter2010b}. The last section briefly concludes. The mathematics of the proof of Theorem \ref{fotetel} is relegated to Appendix \ref{appendix}. 

\section{The type space}\label{typespace}

\textit{Notation}: Let $N$ be the set of the players, w.l.o.g. we can assume that $0 \notin N$, and let $N_0 \circeq N \cup \{0\}$, where $0$ is for the nature as a player.

Let $\# A$ be the cardinality of set $A$. For any set system $\mathcal{A} \subseteq \mathcal{P} (X)$: $\sigma (\mathcal{A})$ is the coarsest $\sigma$-field which contains $\mathcal{A}$. Let $(X,\mathcal{M})$ and $(Y,\mathcal{N})$ be measurable spaces, then $(X \times Y,\mathcal{M} \otimes \mathcal{N})$ or briefly $X \otimes Y$ is the measurable space on the set $X \times Y$ equipped with the $\sigma$-field $\sigma (\{A \times B \mid A \in \mathcal{M}, \ B \in \mathcal{N}\})$.

The measurable spaces $(X,\mathcal{M})$ and $(Y,\mathcal{N})$ are measurable isomorphic if there is a bijection $f$ between them such that both $f$ and $f^{-1}$ are measurable.

For any measurable space $(X,\mathcal{M})$ and for any point $x \in X$: $\delta_x$ is for the Dirac measure on $(X,\mathcal{M})$ concentrated at point $x$.

\bigskip

In the following, we use terminologies that are very similar to \cite{HeifetzSamet1998}'s.

\begin{definition}\label{merhetoseg}
Let $(X,\mathcal{M})$ be a measurable space and denote $\Delta(X,\mathcal{M})$ the set of the probability measures on it. Then the $\sigma$-field $\mathcal{A}^\ast$ on $\Delta(X,\mathcal{M})$ is defined as follows:

\begin{equation*}
\mathcal{A}^\ast \circeq \sigma(\{\{ \mu \in \Delta(X,\mathcal{M}) \mid \mu (A) \geq p\}, \ A \in \mathcal{M}, \ p \in [0,1]\}) \ .
\end{equation*}

\noindent In other words, $\mathcal{A}^\ast$ is the smallest $\sigma$-field among the $\sigma$-fields which contain the sets $\{ \mu \in \Delta(X,\mathcal{M}) \mid \mu (A) \geq p\}$, where $A \in \mathcal{M}$ and $p \in [0,1]$ are arbitrarily chosen.
\end{definition}

In incomplete information situations it is recommended to consider events like a player believes with probability at least $p$ that a certain event occurs (beliefs operator see e.g. \cite{Aumann1999b}). For this reason, for any $A \in \mathcal{M}$ and for any $p \in [0,1]$: $\{ \mu \in \Delta(X,\mathcal{M}) \mid \mu (A) \geq p\}$ must be an event (a measurable set). To keep the class of events as small (coarse) as possible, we use the $\mathcal{A}^\ast$ $\sigma$-field.

Notice that $\mathcal{A}^\ast$ is not a fixed $\sigma$-field, it depends on the measurable space on which the probability measures are defined. Therefore $\mathcal{A}^\ast$ is similar to the $weak^\ast$ topology, which depends on the topology of the base (primal) space.

\begin{assumption}\label{ass1}
Let the parameter space $(S,\mathcal{A})$ be a measurable space.
\end{assumption}

Henceforth we assume that $(S,\mathcal{A})$ is a fixed parameter space which contains all states of the nature. 

\begin{definition}\label{vilagallapot}
Let $\Omega$ be the space of the states of the world and for each $i \in N_0$: let $\mathcal{M}_i$ be a $\sigma$-field on $\Omega$. The $\sigma$-field $\mathcal{M}_i$ represents player $i$'s information, $\mathcal{M}_0$ is for the information available for the nature, hence it is the representative of $\mathcal{A}$, the $\sigma$-field of the parameter space $S$. Let $\mathcal{M} \circeq \sigma (\bigcup \limits_{i \in N_0} \mathcal{M}_i)$, the smallest $\sigma$-field which contains all $\sigma$-fields $\mathcal{M}_i$.
\end{definition}

Each point in $\Omega$ provides a complete description of the actual state of the world. It includes both the state of nature and the players' states of the mind. The different $\sigma$-fields are for modeling the informedness of the players, they have the same role as e.g. the partitions in \cite{Aumann1999a}'s paper have. Therefore, if $\omega,\omega' \in \Omega$ are not distinguishable \footnote{Let $(X,\mathcal{T})$ be a measurable space and $x,y \in X$ be two points. $x$ and $y$ are measurably indistinguishable if $\forall A \in \mathcal{T}$: $(x \in A) \Leftrightarrow (y \in A)$.} in the $\sigma$-field $\mathcal{M}_i$, then player $i$ is not able to discern the difference between them, i.e., she believes the same things and behaves in the same way at the two states $\omega$ and $\omega'$. $\mathcal{M}$ represents all information available in the model, it is the $\sigma$-field got by pooling the information of the players and the nature.

For the sake of brevity, henceforth -- if it does not make confusion -- we do not indicate the $\sigma$-fields. E.g. instead of $(S,\mathcal{A})$ we write $S$, or $\Delta(S)$ instead of $(\Delta(S,\mathcal{A}),\mathcal{A}^\ast)$. However, in some cases we refer to the non-written $\sigma$-field: e.g. $A \in \Delta(X,\mathcal{M})$ is a set of $\mathcal{A}^\ast$, i.e., it is a measurable set in the measurable space $(\Delta(X,\mathcal{M}),\mathcal{A}^\ast)$, but $A \subseteq \Delta(X,\mathcal{M})$ keeps its original meaning: $A$ is a subset of $\Delta(X,\mathcal{M})$.

\begin{definition}\label{tipuster}
Let $(\Omega,\mathcal{M})$ be the space of the states of the world (see Definition \ref{vilagallapot}). The type space based on the parameter space $S$ is a tuple $(S,\{(\Omega,\mathcal{M}_i)\}_{i \in N_0},$ $g,\{f_i\}_{i \in N})$, where

\begin{enumerate}
\item $g : \Omega \rightarrow S$ is $\mathcal{M}_0$-measurable,

\item for each $i \in N$: $f_i : \Omega \to \Delta(\Omega,\mathcal{M}_{-i})$ is $\mathcal{M}_i$-measurable,

\item for each $i \in N$, $\omega \in \Omega$, $A \in \mathcal{M}_{-i}$ such that there exists $A' \in \mathcal{M}_i$, $\omega \in A'$ and $A' \subseteq A$: $f_i (\omega) (A) = 1$,
\end{enumerate}

\noindent  where $\mathcal{M}_{-i} \circeq \sigma(\bigcup \limits_{j \in N_0 \setminus \{i\}} \mathcal{M}_j)$.
\end{definition}

Put Definition \ref{tipuster} differently, $S$ is the parameter space, it contains the "types" of the nature. $\mathcal{M}_i$ represents the information available for player $i$, hence it corresponds to the concept of types (see \cite{Harsanyi196768}). $f_i$ is the type function of player $i$, it assigns player $i$'s (subjective) beliefs to her types.

The above definition of type space differs from \cite{HeifetzSamet1998}'s concept, but it is similar to the type space of \cite{Meier2001} and \cite{Meier2008}. We do not use a Cartesian product space, and refer only to the $\sigma$-fields. By following strictly \cite{HeifetzSamet1998}'s paper, if one takes the Cartesian product of the parameter space and the type sets, and defines the $\sigma$-fields as the $\sigma$-fields induced by the coordinate projections (e.g. $\mathcal{M}_0$ is induced by the coordinate projection $pr_0 : S \times \times_{i \in N} T_i \rightarrow S$, for the notations see their paper), then she gets our concept. However, if the Cartesian product is not used directly, then it is necessary to put the parameter space into the type space in some way. For doing so we use $g$ (\cite{MertensZamir1985} use a similar formalism), hence $g$ and $pr_0$ have the same role in the two formalizations, in ours and in \cite{HeifetzSamet1998} respectively.

A further difference between the two formalizations lies in the role of the parameter space. While in \cite{HeifetzSamet1998} the entire parameter space ''is in'' the space of the states of the world, in our approach that is not required. 
We emphasize that this difference is not relevant.

\begin{definition}\label{tt2}
The mapping $\varphi : \Omega \rightarrow \Omega'$ is a type morphism between type spaces $(S,\{(\Omega,\mathcal{M}_i)\}_{i \in N_0},g,\{f_i\}_{i \in N})$ and $(S,\{(\Omega',\mathcal{M}'_i)\}_{i \in N_0},g',\{f'_i\}_{i \in N})$ if

\begin{enumerate}
\item $\varphi$ is an $\mathcal{M}$-measurable mapping,

\item Diagram \eqref{diag1} is commutative, i.e., for each $\omega \in \Omega$: $g' \circ \varphi (\omega) = g (\omega)$,

\begin{equation}\label{diag1}
\begin{diagram}
\Omega  & & \\
\dTo^\varphi & \rdTo^{g} \\
\Omega' & \rTo^{g'} & S \\
\end{diagram}
\end{equation}

\item for each $i \in N$: Diagram \eqref{morphism} is commutative, i.e., for each $i \in N$, $\omega \in \Omega$: $f'_i \circ \varphi (\omega)  = \hat{\varphi}_i \circ f_i (\omega)$,

\begin{equation}\label{morphism}
\begin{diagram}
\Omega & \rTo^{f_i} & \Delta (\Omega,\mathcal{M}_{-i}) \\
\dTo^{\varphi} & \mspace{200mu} & \dTo^{\hat{\varphi}_i} \\
\Omega' & \rTo^{f'_i} & \Delta (\Omega',\mathcal{M}'_{-i}) \\
\end{diagram}
\end{equation}

\noindent where $\hat{\varphi}_i : \Delta (\Omega,\mathcal{M}_{-i}) \to \Delta (\Omega',\mathcal{M}'_{-i})$ is defined as follows: for all $\mu \in \Delta (\Omega,\mathcal{M}_{-i})$, $A \in \mathcal{M}'_{-i}$: $\hat{\varphi}_i (\mu) (A) = \mu (\varphi^{-1} (A))$. It is a slight calculation to show that $\hat{\varphi}_i$ is a measurable mapping.

\end{enumerate}

\noindent $\varphi $ type morphism is a type isomorphism, if $\varphi $ is a bijection and $\varphi^{-1}$ is also a type morphism.
\end{definition}

The above definition is practically the same as \cite{HeifetzSamet1998}'s, hence all intuitions, they discussed, remain valid, i.e., the type morphism assigns type profiles from a type space to type profiles in an other type space in the way the corresponded types induce equivalent beliefs. In other words, the type morphism preserves the players' beliefs.

The following result is a direct corollary of Definitions \ref{tipuster} and \ref{tt2}.

\begin{corollary}
The type spaces based on the parameter space $S$ as objects and the type morphisms form a category. Let  $\mathcal{C}^S$ denote this category of type spaces.
\end{corollary}

\cite{HeifetzSamet1998} introduce the concept of universal type space.

\begin{definition}\label{tt3}
The type space $(S,\{(\Omega,\mathcal{M}_i)\}_{i \in N_0},g,\{f_i\}_{i \in N})$ is a universal type space, if for any type space $(S,\{(\Omega',\mathcal{M}'_i)\}_{i \in N_0},g',\{f'_i\}_{i \in N})$ there exists a unique type morphism $\varphi : \Omega' \to \Omega$.
\end{definition}

In other words, the universal type space is the most general, the broadest type space among the type spaces. It contains all types which appear in the type spaces of the given category.

In the language of category theory Definition \ref{tt3} means the following:

\begin{corollary}
The universal type space is a terminal (final) object in category $\mathcal{C}^S$.
\end{corollary}

From the viewpoint of category theory the uniqueness of a universal type space is a straightforward statement.

\begin{corollary}\label{egyetlen}
The universal type space is unique up to type isomorphism. 
\end{corollary}

\begin{proof}
Every terminal object is unique up to isomorphism.
\end{proof}

The only question is the existence of a universal type space. Next, we present a result which is an adaptation of \cite{HeifetzSamet1998} Theorem 3.4 to our setting.

\begin{proposition}\label{egyetemes}
There exists a universal type space, in other words, there is a terminal object in category $\mathcal{C}^S$.
\end{proposition}

As we have already mentioned, \cite{HeifetzSamet1998}'s formalization of type spaces is different from ours. Since the difference between the two approaches is quite slight and we prove a stronger result in Theorem \ref{fotetel}, we omit the proof of the above proposition.

Next, we turn our attention to an other property of type spaces, to the completeness.

\begin{definition}\label{teljesseg}
The type space $(S,\{(\Omega,\mathcal{M}_i)\}_{i \in N_0},g,\{f_i\}_{i \in N})$ is complete, if for each $i \in N$: $f_i$ is surjective (onto).
\end{definition}

\cite{Brandenburger2003} introduces the concept of complete type space. The completeness recommends that for any player $i$, any probability measure on $(\Omega,$ $\mathcal{M}_{-i})$ be in the range of the given player's type function. In other words, for any player $i$, any measure on $(\Omega,\mathcal{M}_{-i})$ must be assigned (by the type function $f_i$) to a type of player $i$. The following proposition is practically the same as \cite{Meier2001}'s Theorem 4.

\begin{proposition}\label{fontosallitas}
The universal type space is complete.
\end{proposition}

Although \cite{Meier2001}'s type space is different from ours, the difference is slight, moreover, we prove a stronger result in Theorem \ref{fotetel}, hence we also omit the proof of the above proposition.

\section{The beliefs space}\label{bs}

In this section we formalize the intuition of hierarchies of beliefs, as \cite{Harsanyi196768} calls them the "infinite regress in reciprocal expectations." First, we give a rough description (see \cite{MertensZamir1985}):

\begin{equation*}\label{alap}
\begin{array}{lcl}
T_0 & \circeq & S \\
T_1 & \circeq & T_0 \otimes \Delta(T_0)^N \\
T_2 & \circeq & T_1 \otimes \Delta(T_1)^N = T_0 \otimes \Delta(T_0)^N \otimes \Delta(T_0 \otimes \Delta(T_0)^N)^N \\
& \vdots & \\
T_n & \circeq & T_{n-1} \otimes \Delta(T_{n-1})^N = T_0 \otimes \bigotimes \limits_{m=0}^{n-1} \Delta(T_m)^N \\
& = & T_0 \otimes \bigotimes \limits_{m=0}^{n-2} \Delta(T_m)^N \otimes \Delta(T_0 \otimes \bigotimes \limits_{m=0}^{n-2} \Delta(T_m)^N)^N \\
& \vdots &
\end{array}
\end{equation*}

The above formalism can be interpreted as follows. $T_0$ describes the basic uncertainty of the modeled situation, its elements are the states of nature. $T_1$ is for $T_0$ and the first order beliefs of the players $\Delta (T_0)^N$, i.e., the players' beliefs about the states of nature. In general, $T_n$ describes $T_{n-1}$ and the $n$th order beliefs of the payers $\Delta (T_{n-1})^N$, i.e., the players' beliefs about $T_{n-1}$.

There are, however, some redundancies in this model. First, \cite{Harsanyi196768} proposes that the players know their own types (see point $3$. in Definition \ref{tipuster}), so e.g. player $i$'s belief about her own first order belief is a Dirac measure. Second, there is an other redundancy\footnote{This redundancy is called \textit{coherency} and \textit{consistency} in the literature of game theory and mathematics respectively.} in the above description as well. E.g. $\Delta(T_0 \otimes \Delta(T_0)^N)^N$ determines $\Delta (T_0)^N$ and so does $\Delta(T_{n-1})^N$ for all $0 \leq m \leq n-2$: $\Delta(T_m)^N$.

Therefore we can rewrite the above formalism, from the viewpoint of player $i$ as follows (the definition below is a purely measurable reformulation of \cite{MertensSorinZamir1994}'s concept):

\begin{definition}\label{velemenyter}
In Diagram \eqref{limit2}

\begin{equation}\label{limit2}
\begin{diagram}
\Theta^i & \mspace{100mu} & \Delta (S \otimes \Theta^{N \setminus \{i\}}) \\
\dTo^{p^i_{n+1}} & & \dTo^{id_S} \mspace{35mu} \dTo_{p^{N \setminus \{i\}}_n} \\
\Theta_{n+1}^i & \circeq & \Delta (S \otimes \Theta_{n}^{N \setminus \{i\}}) \\
\dTo^{q_{nn+1}^i} & & \dTo^{id_S} \mspace{35mu} \dTo_{q_{n-1n}^{N \setminus \{i\}}} \\
\Theta_{n}^i & \circeq & \Delta (S \otimes \Theta_{n-1}^{N \setminus \{i\}}) \\
\end{diagram}
\end{equation}

\begin{itemize}
\item $i \in N$ is a player,

\item $n \in \mathbb{N}$,

\item $S$ is the parameter space (see Assumption \ref{ass1}).
\end{itemize}

Moreover for each $i \in N$: 

\begin{itemize}
\item $\# \Theta_{-1}^i = 1$,

\item for each $n \in \mathbb{N} \cup \{-1\}$: $\Theta_n^{N \setminus \{i\}} \circeq \bigotimes \limits_{j \in {N \setminus \{i\}}} \Theta^{j}_n$,

\item $q_{-10}^i : \Theta_0^i \rightarrow \Theta_{-1}^i$,

\item for all $m,n \in \mathbb {N}$ such that $m \leq n$, $\mu \in \Theta_n^i$:

\begin{equation*}
q_{mn}^i (\mu) \circeq \mu|_{S \otimes \Theta_{m-1}^{N \setminus \{i\}}} \ ,
\end{equation*}

\noindent therefore $q_{mn}^i$ is a measurable mapping.

\item $\Theta^i \circeq  \varprojlim (\Theta_n^i,\mathbb{N} \cup \{-1\},q_{mn}^i)$,

\item for each $n \in \mathbb{N} \cup \{-1\}$: $p_n^i : \Theta^i \rightarrow \Theta_{n}^i$ is the canonical projection,

\item for all $m,n \in \mathbb{N} \cup \{-1\}$ such that $m \leq n$: $q_{mn}^{N \setminus \{i\}}$ is the product of the mappings $q_{mn}^j$, $j \in {N \setminus \{i\}}$, and so is $p_n^{N \setminus \{i\}}$ of $p_n^j$, $j \in {N \setminus \{i\}}$, therefore both mappings are measurable,

\item $\Theta^{N \setminus \{i\}} \circeq  \bigotimes \limits_{j \in {N \setminus \{i\}}} \Theta^{j}$.
\end{itemize}

\noindent Then $T \circeq S \otimes \Theta^N$ is called a purely measurable beliefs space.
\end{definition}

The interpretation of the purely measurable beliefs space is the following. For any $\theta^i \in \Theta^i$: $\theta^i = (\mu_{1}^i,\mu_{2}^i,\ldots)$, where $\mu_{n}^i \in \Theta_{n-1}^i$ is player $i$'s $n$th order belief. Therefore each point of $\Theta^i$ defines an inverse system of probability measure spaces 

\begin{equation}\label{hb}
((S \otimes \Theta_n^{N \setminus \{i\}},p_{n+1}^i (\theta^i)),\mathbb{N} \cup \{-1\},(id_S,q_{mn}^{N \setminus \{i\}})) \ ,
\end{equation}

\noindent where $(id_S,q_{mn}^{N \setminus \{i\}})$\footnote{It is clear that this is a measurable mapping.} is the product of mappings $id_S$ and $q_{mn}^{N \setminus \{i\}}$. We call the inverse systems of probability measure spaces like \eqref{hb} player $i$'s \textit{hierarchies of beliefs}\footnote{In the literature system like \eqref{hb} is usually called coherent hierarchy of beliefs. Since it does not make confusion in this paper we omit the adjective coherent.}.

To sum up, $T$ consists of all states of the world: all states of nature: the points of $S$, and all players' all states of the mind: the points of set $\Theta^N$, therefore $T$ contains all players' all hierarchies of beliefs.

Our main result:

\begin{theorem}\label{fotetel}
The complete universal type space contains all players' all hierarchies of beliefs.
\end{theorem}

We present the proof of Theorem \ref{fotetel} in the next section.

\section{The proof of Theorem \ref{fotetel}}\label{proof}

The strategy of the proof is to show that the purely measurable beliefs space (see Definition \ref{velemenyter}) generates (is equivalent to) the complete universal type space (in category $\mathcal{C}^S$). Mathematically, the key point of the proof is to demonstrate the following lemma:

\begin{lemma}\label{fontos2}
For each $i \in N$, $\theta^i \in \Theta^i$: the inverse system of probability measure spaces

\begin{equation}
((S \otimes \Theta_n^{N \setminus \{i\}},p_{n+1}^i (\theta^i)),\mathbb{N} \cup \{-1\},(id_S,q_{mn}^{N \setminus \{i\}}))
\end{equation}

\noindent admits a unique inverse limit.
\end{lemma}

Since the proof of Lemma \ref{fontos2} is technical, we relegated it to Appendix \ref{appendix}.

It is a slight calculation to see that Lemma \ref{fontos2} implies directly that for each $i \in N$ in Diagram \eqref{limit2} 

\begin{equation*}\label{kell0}
\Theta^i = \Delta (S \otimes \Theta^{N \setminus \{i\}}) \ ,
\end{equation*}

\noindent i.e., those are measurable isomorphic. 

First we show that the beliefs space of Definition \ref{velemenyter} induces a type space.

\begin{lemma}\label{indukalttipuster}
The purely measurable beliefs space $T$ induces a type space in category $\mathcal{C}^S$.
\end{lemma}

\begin{proof}
For each $i \in N$: let $pr_i : T \rightarrow \Theta^{i}$, $pr_0 : T \rightarrow S$ be coordinate projections, and for each $i \in N \cup \{0\}$: let the $\sigma$-fields $\mathcal{M}_i^\ast$ be induced by $pr_i$. From Lemma \ref{fontos2} for each $i \in N$:

\begin{equation}\label{fontos}
\Theta^i = \Delta (S \otimes \Theta^{N \setminus \{i\}}) \ ,
\end{equation}

\noindent i.e., those are measurable isomorphic.

Furthermore, let $g^\ast \circeq pr_0$ and for each $t \in T$: $f_i^\ast (t) \circeq pr_i (t)$. Then

\begin{equation*}
(S,\{(T,\mathcal{M}_i^\ast)\}_{i \in N},g^\ast,\{f_i^\ast\}_{i \in N}) 
\end{equation*}

\noindent is a type space in category $\mathcal{C}^S$.
\end{proof}

The following proposition is a direct corollary of identity \eqref{fontos}.

\begin{proposition}\label{teljeskell2}
The type space $(S,\{(T,\mathcal{M}_i^\ast)\}_{i \in N},g^\ast,\{f_i^\ast\}_{i \in N})$ is complete.
\end{proposition}

Next we show that the type space induced by the purely measurable beliefs space is the universal type space.

\begin{proposition}\label{atetel2}
The type space $(S,\{(T,\mathcal{M}_i^\ast)\}_{i \in N},g^\ast,\{f_i^\ast\}_{i \in N})$ is a universal type space in category $\mathcal{C}^S$.
\end{proposition}

\begin{proof}
Let $(S,\{(\Omega,\mathcal{M}_i)\}_{i \in N},g,\{f_i\}_{i \in N})$ be a type space (an object in $\mathcal{C}^S$), $i \in N$ and $\omega \in \Omega$.

Player $i$'s first order belief at state of the world $\omega$ $v_1^i (\omega)$ is the probability measure defined as follows, for each $A \in S$:

\begin{equation*}
v_1^i (\omega) (A) \circeq f_i (\omega) (g^{-1} (A)) \ .
\end{equation*}

\noindent $f_i$ is $\mathcal{M}_i$-measurable, hence $v_1^i$ is also $\mathcal{M}_i$-measurable.

Player $i$'s second order belief at state of the world $\omega$ $v_2^i (\omega)$ is the probability measure defined as follows, for each $A \in S \otimes \Theta_0^{N \setminus \{i\}}$:

\begin{equation*}
v_2^i (\omega) (A) \circeq f_i (\omega) ((g,v_1^{N \setminus \{i\}})^{-1} (A)) \ ,
\end{equation*} 

\noindent where for each $\omega'$: $(g,v_1^{N \setminus \{i\}}) (\omega') \circeq (g(\omega'),\{v_1^j (\omega')\}_{j \in N \setminus \{i\}})$, hence $(g,v_1^{N \setminus \{i\}})$ is $\mathcal{M}_{-i}$-measurable. Since $f_i$ is $\mathcal{M}_i$-measurable $v_2^i$ is also $\mathcal{M}_i$-measurable.

For any $n > 1$ player $i$'s $n$th order belief at state of the world $\omega$  $v_n^i (\omega)$ is the probability measure defined as follows, for each $A \in S \otimes \Theta_{n-2}^{N \setminus \{i\}}$:

\begin{equation*}
v_n^i (\omega) (A) \circeq f_i (\omega) ((g,v_{n-1}^{N \setminus \{i\}})^{-1} (A)) \ .
\end{equation*} 

\noindent Since $f_i$ is $\mathcal{M}_i$-measurable $v_n^i$ is also $\mathcal{M}_i$-measurable.

To sum up, we have got a mapping $\phi : \Omega \rightarrow T$ defined as follows, for each $\omega \in \Omega$:

\begin{equation}\label{psidef}
\phi (\omega) \circeq (g(\omega),(v^i_1(\omega),v^i_2 (\omega),\ldots)_{i \in N}) \ . 
\end{equation}

Then it is easy to verify that

(1) $\phi$ is $\mathcal{M}$-measurable,

(2) for each $i \in N$, $\omega \in \Omega$: 

\begin{equation*}
g^\ast \circ \phi (\omega) = g (\omega)
\end{equation*}

\noindent and

\begin{equation*}
f^\ast_i \circ \phi (\omega)  = \hat{\phi}_i \circ f_i (\omega)\ ,
\end{equation*}

\noindent i.e., $\phi$ is a type morphism,

(3) Since $\Theta^i$ consists of different inverse systems of probability measure spaces (hierarchies of beliefs), $\phi$ is the unique type morphism from the type space $(S,\{(\Omega,$ $\mathcal{M}_i)\}_{i \in N},g,\{f_i\}_{i \in N})$ to $(S,\{(T,\mathcal{M}_i^\ast)\}_{i \in N},g^\ast,\{f_i^\ast\}_{i \in N})$.
\end{proof}

In the above proof we show that any point in a type space induces a hierarchy of beliefs for each player, i.e., any point in a type space completely describes the players' hierarchies of beliefs at the given state of the world. \cite{BattigalliSiniscalchi1999} provided this observation first in the literature. 

It is also worth noticing that in the above proof $\phi$ is not necessarily injective (one-to-one). The $\phi$-image of redundant types, i.e., types such that generate the same hierarchy of beliefs, see e.g. \cite{ElyPeski2006}, is one point in the universal type space. Therefore, it is not surprising at all that there are no redundant types in the universal type space, hence it can be complete.

\begin{proof}[The proof of Theorem \ref{fotetel}]
From Proposition \ref{atetel2}

\begin{equation}\label{univ}
(S,\{(T,\mathcal{M}_i)\}_{i \in N},g^\ast,\{f_i^\ast\}_{i \in N})
\end{equation}

\noindent is the universal type space.

Then Corollary \ref{egyetlen} implies that \cite{HeifetzSamet1998}'s universal type space and \eqref{univ} coincide (those are type isomorphic).

From Proposition \ref{teljeskell2}: \eqref{univ} is complete, \cite{Meier2001} also proves this result.

Finally, from Definition \ref{velemenyter}: \eqref{univ} contains all players' all hierarchies of beliefs.
\end{proof}

\section{Related papers}\label{comparison}

In this section our main result, Theorem \ref{fotetel}, is compared to the results of \cite{HeifetzSamet1999} and \cite{Pinter2010b}. These papers seem to contradict our main result, in the following, however, we show that this is not the case at all.

\bigskip

\cite{HeifetzSamet1999} as the title indicates, give an example for a hierarchy of beliefs, which can not be a type in any type space. Mathematically, their counterexample is based on an exercise of \cite{Halmos1974}, an example for an inverse system of probability measure spaces having no inverse limit. First, we summarize \cite{HeifetzSamet1999}'s example\footnote{We do not follow \cite{HeifetzSamet1999} letter by letter, we only grab the very essence of their example.}.

\begin{example}\label{plHS}
Notation: $l^\ast$ and $l_\ast$ are respectively the outer and inner measures induced by the Lebesgue measure. Let $\{A_n\}_n$ be the Vitali sets from the example of \cite{Halmos1974}, so it is true that for all $n$ $A_{n+1} \subseteq A_n \subseteq A_0 =  [0,1]$, $l^\ast(A_n) = 1$, $l_\ast (A_n) = 0$ and $\bigcap \limits_n A_n = \emptyset$. Moreover, for each $n$ let $\mu_n$ be the probability measures on $B \left( \prod \limits_{k=0}^n A_k \right)$\footnote{$B(\cdot)$ is for the Borel $\sigma$-field.} also from \cite{Halmos1974}'s example.

Consider the following inverse system of probability measure spaces:

\begin{equation}\label{kell}
\left( \left( \prod \limits_{k=0}^n A_k,B \left( \prod \limits_{k=0}^n A_k \right),\mu_n \right),\mathbb{N},pr_{mn} \right) \ ,
\end{equation}

\noindent where $pr_{mn} : \prod \limits_{k=0}^n A_k \rightarrow \prod \limits_{k=0}^m A_k$ is the coordinate projection.

Furthermore, if $X \circeq \prod \limits_{k=0}^n X_k$ is a product space and $\delta_x$ is the Dirac measure concentrated at $x \circeq (x_0,x_1,\ldots,x_n)$, then $\delta_x = \prod \limits_{k=0}^n \delta_{x_k}$, where $\delta_{x_k}$ is the Dirac measure on $B(X_k)$ concentrated at point $x_k$.

Interpretation: There are two players, we chose one of them. Let $A_0$ be the parameter space (the set of the states of nature). $A_1 \subseteq A_0$, and for each $x \in A_1$ let $x$ be $\delta_{x}$, i.e., $A_1$ is\footnote{Henceforth in the context like this "is" means that the two spaces are homeomorphic.} the set of some first order beliefs of the given player. Moreover, $A_2 \subseteq A_1$ and for all $x \in A_2$ let $x$ be $\delta^2_{x}$, where $\delta^2_x \circeq \delta_{\delta_{x}}$, i.e., $A_1 \times A_2$ is the set of some second order beliefs of the given player. In general, for each $n \geq 3$, $x \in A_n \subseteq A_{n-1}$: let $x$ be $\delta^n_x$, where $\delta^n_{x} \circeq \delta_{\delta^{n-1}_{x}}$, i.e., $\prod \limits_{k=1}^n A_k$ is the set of some $n$th order beliefs of the given player.

To sum up, for each $n$ $(a_0,a_1,\ldots,a_n) \in \prod \limits_{k=0}^n A_k$ is $(a_0,\delta_{a_1},\delta^2_{a_2},\ldots,\delta^n_{a_n}) = a_0 \times \delta_{(a_1,\ldots,a_n)}$.

Put it differently, $\prod \limits_{k=1}^\infty A_k$ is the space of some hierarchies of beliefs, therefore \eqref{kell} is a hierarchy of beliefs. However, from the example of \cite{Halmos1974} \eqref{kell} has no inverse limit, i.e., this hierarchy of beliefs is not a type.
\end{example}

Next we show that \cite{HeifetzSamet1999}'s hierarchy of beliefs is not in the purely measurable beliefs space $T$ (see Definition \ref{velemenyter}).

\begin{lemma}
\eqref{kell} is not in $T$.
\end{lemma}

\begin{proof}
It is enough to show that the diagonal of $A_0 \times A_1$ is not a measurable subset of $A_0 \otimes \Delta (A_0)$. The strategy of the proof is the following: if the diagonal of $A_0 \times A_1$ is a measurable subset of $A_0 \otimes \Delta (A_0)$, then for any $B \subseteq A_0 \otimes \Delta (A_0)$ the intersection of the diagonal of $A_0 \times A_1$ and $B$ is a measurable set of subspace $B$.

Consider $A_0 \times A_0$, i.e., let $B \circeq [0,1] \times [0,1]$. Then from Example \ref{plHS} $B$ equipped with the subspace $\sigma$-field is measurable isomorphic to $A_0 \otimes \Delta_D (A_0)$, where $\Delta_D (A_0)$ is for the set of Dirac measures on $A_0$. Furthermore, it is clear that $B([0,1] \times [0,1])$ encompasses the subspace $\sigma$-field of $B$. However, from the definition of $\{A_n\}_n$ the diagonal of $A_0 \times A_1$ is not a Borel measurable subset of (the diagonal of) $[0,1] \times [0,1]$, i.e., it is not in the subspace $\sigma$-field of $B$, hence $\mu_1 \notin \Delta (A_0 \otimes \Delta (A_0))$, where $\mu_1$ is from Example \ref{plHS}.
\end{proof}

Even if we have showed above that \cite{HeifetzSamet1999}'s hierarchy of beliefs is not in the purely measurable beliefs space, they argue that it is a purely measurable hierarchy of beliefs, so we have to comment their argument. Without going into the details, we can say that the ''problem'' with \cite{HeifetzSamet1999}'s argument is the following: the Cartesian production is not commutative, i.e., e.g. (see their notations\footnote{$S = A_0 \times A_1 \times \cdots \times A_n \times \cdots$.})

\begin{equation*}
S^2 \neq  A_0 \times A_0 \times A_1 \times A_1 \times A_2 \times \cdots \times A_n \times A_{n+1} \times \cdots 
\end{equation*} 

This makes trouble because at certain points \cite{HeifetzSamet1999} use $S^2$ at other points they use $A_0 \times A_0 \times A_1 \times A_1 \times A_2 \times \cdots$. The first is needed for getting an inverse system of probability measure spaces having no inverse limit, so is the second for embedding the inverse system of probability measure spaces under consideration into the purely measurable beliefs space. 

It is easy to see that (e.g.) the notion of diagonal depends on the ordering, e.g. $\textup{Diag} ((A \times B) \times (A \times B)) \neq \textup{Diag} ((B \times A) \times (A \times B))$, which means a crucial flaw in \cite{HeifetzSamet1999}'s argument (they need the diagonal of $S^2$).

To sum up, \cite{HeifetzSamet1999}'s counterexample is a hierarchy of beliefs such that it is not among the purely measurable hierarchies of beliefs, i.e., it is not in the purely measurable beliefs space. Therefore, \cite{HeifetzSamet1999}'s result does not contradict ours.

\bigskip

Very recently, \cite{Pinter2010b} provides a negative result, he argues that there is no universal topological type space in the category of topological type spaces. Actually, this non-existence is got by a topological reasoning, hence this negative result does not contradict this paper's positive one.

On the other hand, \cite{Pinter2010b}'s result clearly shows that the irrelevant details, getting in the model by topological concepts, can make real difficulties, which culminate in that the goal proving that the Harsányi program works, is unreachable in the topological setting.

\section{Conclusion}

The main result of this paper, Theorem \ref{fotetel}, concludes that in the purely measurable framework the Harsányi program works, i.e., the incomplete information situations can be modeled by type spaces.

Theorem \ref{fotetel} together with \cite{Pinter2010b}'s result raise the problem that although in the literature mostly the topological models are popular, the purely measurable and not the topological framework is appropriate for modeling incomplete information situations. Can the results of the topological framework be translated into the purely measurable one? For this question future research can answer.

\appendix

\section{The proof of Lemma 5.1}\label{appendix}

In this appendix we prove that $\forall i \in N$, $\forall \theta^i \in \Theta^i$: the inverse system of probability measure spaces

\begin{equation}\label{hba}
((S \otimes \Theta_n^{N \setminus \{i\}},p_{n+1}^i (\theta^i)),\mathbb{N} \cup \{-1\},(id_S,q_{mn}^{N \setminus \{i\}}))
\end{equation} 

\noindent admits a unique inverse limit.

First we refer to a concept and a result, both from \cite{Pinter2010a}.

\begin{definition}\label{eros}
The inverse system of probability measure spaces $((X_n,$ $\mathcal{M}_n,$ $\mu_n),\mathbb{N},f_{mn})$ is $\varepsilon$-complete, if $\forall \varepsilon \in [0,1]$, $\forall m,n \in \mathbb{N}$ such that $m \leq n$, and $\forall A \subseteq X_{m}$:

\begin{equation*}
(\mu_n^\ast (f_{mn}^{-1} (A)) = \varepsilon) \Rightarrow (\mu_m^\ast (A) = \varepsilon) \ ,
\end{equation*}

\noindent where $\mu_n^\ast$ is the outer measure generated by $\mu_n$. 
\end{definition}

The following result is Theorem 3.2 from \cite{Pinter2010a}.

\begin{theorem}\label{lemmakell1}
Let $((X_n,\mathcal{M}_n,\mu_n),\mathbb{N},f_{mn})$ be an $\varepsilon$-complete inverse system of probability measure spaces. Then $(X,\mathcal{M},\mu) \circeq \varprojlim ((X_n,\mathcal{M}_n,\mu_n),\mathbb{N},f_{mn})$ exists and is unique.
\end{theorem}

\begin{remark}
It is worth noticing that in Theorem \ref{lemmakell1} we can substitute fields for the $\sigma$-fields $\mathcal{M}_n$.
\end{remark}

\begin{remark}\label{lemma1}
Let $(X,\mathcal{M},\mu)$ and $(Y,\mathcal{N},\nu)$ be probability measure spaces, and $f: X \to Y$ be a measurable mapping. Then the following two conditions are equivalent

\begin{equation*}
\inf \limits_{B \in \mathcal{M},\ f^{-1} (A) \subseteq B} \mu (B) = \inf \limits_{B \in \mathcal{N},\ A \subseteq B} \nu (B)\ , \mspace{50mu} A \subseteq Y\ ,
\end{equation*}

\noindent and

\begin{equation*}
\sup \limits_{B \in \mathcal{M},\ B \subseteq f^{-1} (A)} \mu (B) = \sup \limits_{B \in \mathcal{N},\ B \subseteq A} \nu (B)\ , \mspace{50mu} A \subseteq Y\ .
\end{equation*}
\end{remark}

Next, we give an "alternative" definition of $\varepsilon$-completeness.

\begin{corollary}\label{kov1}
The inverse system of probability measure spaces $((X_n,$ $\mathcal{M}_n,\mu_n),$ $\mathbb{N},f_{mn})$ is $\varepsilon$-complete if and only if $\forall m,n \in \mathbb{N}$ such that $m \leq n$, and $\forall A \subseteq X_m$:

\begin{equation*}
\sup \limits_{B \in \mathcal{M}_n,\ B \subseteq f^{-1}_{mn} (A)} \mu_n (B) = \sup \limits_{B \in \mathcal{M}_m,\ B \subseteq A} \mu_m (B) \ .
\end{equation*}
\end{corollary}

For the sake of simplicity, henceforth, we assume that there are only two players, and we consider the problem from the viewpoint of one of them. Then we get the following diagram (see Diagram \eqref{limit2}):

\begin{equation}\label{limit2a}
\begin{diagram}
\Theta & \mspace{100mu} & \Delta (S \otimes \Theta) \\
\dTo^{p_{n+1}} & & \dTo^{id_S} \mspace{35mu} \dTo_{p_n} \\
\Theta_{n+1} & \circeq & \Delta (S \otimes \Theta_{n}) \\
\dTo^{q_{nn+1}} & & \dTo^{id_S} \mspace{35mu} \dTo_{q_{n-1n}} \\
\Theta_{n} & \circeq & \Delta (S \otimes \Theta_{n-1}) \\
\end{diagram}
\end{equation}

Since here we consider the case where there are only two players we focus on the following inverse system of probability measure spaces:

\begin{equation}\label{hb0}
((S \times \Theta_n,p_{n+1}),\mathbb{N} \cup \{-1\},(id_S,q_{mn})) \ .
\end{equation} 

Consider the ''truncated'' variant of \eqref{hb0}:

\begin{equation}\label{hb1}
((\Theta_n,p_{n+1} (\theta)|_{\Theta_n}),\mathbb{N} \cup \{-1\},q_{mn}) \ .
\end{equation} 

The following lemma is a direct corollary of \cite{MarczewskiRyll1953}'s result (see e.g. \cite{Bogachev2006b} Problem 7.14.100 on p. 161).

\begin{lemma}\label{MRN}
Let $(X,\mathcal{M})$ be a measurable space and $\mathcal{A} \subseteq \mathcal{P} (X)$ be a field such that $\# \mathcal{A} < \infty$. Moreover, let $\mu$ be an additive probability set function on the field generated by $\mathcal{A} \cup \mathcal{M}$ such that $\mu|_{(X,\mathcal{M})}$, the restriction of $\mu$ onto the measurable space $(X,\mathcal{M})$ is $\sigma$-additive. Then $\mu$ is $\sigma$-additive.
\end{lemma}

Notice that for any measurable space $(X,\mathcal{M})$, $\Delta (X,\mathcal{M}) \subset [0,1]^\mathcal{M}$. Therefore, $\Delta (X,\mathcal{M})$ can be equipped with the pointwise convergence topology as a subspace of $[0,1]^\mathcal{M}$. Henceforth, let $\tau_p$ denote the pointwise convergence topology.

The next lemma shows that the inverse system of probability measure spaces \eqref{hb1} can be "embedded" uniquely into the following inverse system of probability measure spaces:

\begin{equation}\label{hb11}
((\Theta_n,B(\Theta_n,\tau_p),\mu_n),\mathbb{N} \cup \{-1\},q_{mn}) \ ,
\end{equation} 

\noindent where $B(\Theta_n,\tau_p)$ is for the Borel $\sigma$-field of the pointwise convergence topology and $\mu_n$ is an inner $\tau_p$-closed regular measure such that $\mu_n|_{\Theta_n} = p_{n+1} (\theta)|_{\Theta_n}$.
 
\begin{lemma}\label{fontos4}
Let $\mu$ be a probability measure on $(\Delta (X,\mathcal{M}),\mathcal{A}^\ast)$, where $\mathcal{A}^\ast$ is given in Definition \ref{merhetoseg}. Moreover, let $\tau_p$ be the pointwise convergence topology on $\Delta (X,\mathcal{M})$. Then there exists a unique probability measure $\nu$ on $B(\Delta (X,\mathcal{M}),\tau_p)$ such that 

\begin{enumerate}
\item  $\mu = \nu|_{(\Delta (X,\mathcal{M}),\mathcal{A}^\ast)}$,

\item  $\nu$ is inner $\tau_p$-closed regular, i.e., $\forall A \in B(\Delta (X,\mathcal{M}),\tau_p)$, $\forall \varepsilon > 0$: $\exists Z$ $\tau_p$-closed set such that $Z \subseteq A$ and $\nu (A \setminus Z) < \varepsilon$.
\end{enumerate}
\end{lemma}

\begin{proof}
First, from Definition \ref{merhetoseg} $\forall A \in \mathcal{A}$, $\forall \varepsilon > 0$: $\exists Z \in \mathcal{A}^\ast$ $\tau_p$-closed set such that $Z \subseteq A$ and $\mu (A \setminus Z) < \varepsilon$. 

Let $a (X,\mathcal{M})$ be the set of the additive probability set functions on the measurable space $(X,\mathcal{M})$. Notice that $a (X,\mathcal{M})$ is a compact topological space with the pointwise convergence topology. Let 

\begin{equation*}
\mathcal{A}^{\ast \ast} \circeq \sigma(\{\{ \mu \in a(X,\mathcal{M}) \mid \mu (A) \geq p\}, \ A \in \mathcal{M}, \ p \in [0,1]\}) \ .  
\end{equation*}

\noindent Then $\forall A \subseteq \mathcal{A}^{\ast \ast}$: $i^{-1} (A) \in \mathcal{A}^\ast$, where $i : \Delta (X,\mathcal{M}) \to a (X,\mathcal{M})$ is the natural embedding; moreover let $\mu^\ast \circeq \mu \circ i^{-1}$ be a probability measure on $(a (X,\mathcal{M}),$ $\mathcal{A}^{\ast \ast})$. Notice that  $\forall A \in \mathcal{A}^{\ast \ast}$, $\forall \varepsilon > 0$: $\exists Z \in \mathcal{A}^{\ast \ast}$ compact set in the pointwise convergence topology such that $Z \subseteq A$ and $\mu^\ast (A \setminus Z) < \varepsilon$.

Furthermore, notice that $\mathcal{A}^{\ast \ast}$ contains the base of the pointwise convergence topology on $a (X,\mathcal{M})$, hence Henry's Extension Theorem (see e.g. \cite{Rao1987} Theorem 10 on pp. 76-78 and Exercise 11 on p. 80) implies that there exists a unique probability measure $\nu^\ast$ on the Borel $\sigma$-field of the pointwise convergence topology on $a (X,\mathcal{M})$ such that $\mu^\ast = \nu^\ast_{(a (X,\mathcal{M}),\mathcal{A}^\ast)}$, and $\nu^\ast$ is inner compact regular, i.e.,  for any set $A$ of the Borel $\sigma$-field of the pointwise convergence topology on $a (X,\mathcal{M})$, $\forall \varepsilon > 0$: $\exists Z$ compact set in the pointwise convergence topology such that $Z \subseteq A$ and $\nu^\ast (A \setminus Z) < \varepsilon$.

Let $\nu$ be defined by $\nu \circ i^{-1} \circeq \nu^\ast$; then $\nu$ is a probability measure on the Borel $\sigma$-field  of the pointwise convergence topology on $\Delta (X,\mathcal{M})$, $\mu = \nu|_{(\Delta (X,\mathcal{M}),\mathcal{A}^\ast)}$, and $\nu$ is inner $\tau$-closed regular, i.e., $\forall A \in B(\Delta (X,\mathcal{M}),\tau)$, $\forall \varepsilon > 0$: $\exists Z$ $\tau$-closed set such that $Z \subseteq A$ and $\nu (A \setminus Z) < \varepsilon$.
\end{proof}

The next lemma demonstrates that the mappings $q_{mn}$ in \eqref{hb11} are closed mappings.

\begin{lemma}\label{fontos3}
Let $(X,\mathcal{M}_1)$ and $(X,\mathcal{M}_2)$ be measurable spaces such that $\mathcal{M}_1 \subseteq \mathcal{M}_2$. Moreover, let $\Delta(X,\mathcal{M}_1)$ and $\Delta(X,\mathcal{M}_2)$ be equipped with the pointwise convergence topologies $\tau_1$ and $\tau_2$ respectively, and let $f : \Delta(X,$ $\mathcal{M}_2) \to \Delta (X,$ $\mathcal{M}_1)$ be defined as follows, $\forall \nu \in \Delta(X,\mathcal{M}_2)$: $f(\nu) \circeq \nu|_{(X,\mathcal{M}_1)}$. Then $f$ is a closed-mapping, i.e., for any $\tau_2$-closed set $Z$: $f(Z)$ is $\tau_1$-closed. 
\end{lemma}

\begin{proof}
In Diagram \ref{komp} $a(X,\cdot)$ is for the set of the additive probability set functions on $(X,\cdot)$, $i_1$ and $i_2$ are for the natural injections (embeddings), and $f_a : a (X,\mathcal{M}_2) \to a (X,\mathcal{M}_1)$ is defined as follows: $\forall \nu \in a (X,\mathcal{M}_2)$: $f_a (\nu) = \nu|_{(X,\mathcal{M}_1)}$.

\begin{equation}\label{komp}
\begin{diagram}
\Delta (X,\mathcal{M}_2) &\rTo^f & \Delta (X,\mathcal{M}_1) \\
\dTo^{i_2} &  \mspace{200mu} & \dTo^{i_1} \\
a (X,\mathcal{M}_2) & \rTo^{f_a} & a (X,\mathcal{M}_1)
\end{diagram}
\end{equation}

First notice that the both $a (X,\mathcal{M}_1)$ and $a (X,\mathcal{M}_2)$ are compact topological spaces (with the pointwise convergence topology), and $f$, $f_a$, $i_1$ and $i_2$ are continuous mappings.

Moreover, notice that for any $\tau_2$-closed set $Z$: 

\begin{equation}\label{kell10}
f_a (\overline{i_2 (Z)}) = \overline{i_1 (f (Z))} \ ,
\end{equation}

\noindent where $\overline{i_2 (Z)}$ is for the pointwise convergence topology closure of set $Z$ in $a (X,$ $\mathcal{M}_2)$ and $\overline{i_1 (f (Z))}$ is for the the pointwise convergence topology closure of set $f(Z)$ in $a (X,\mathcal{M}_1)$.

Let $Z$ be a $\tau_2$-closed set, and indirectly assume that $f(Z)$ is not $\tau_1$-closed. Let $\mu \in \Delta(X,\mathcal{M}_1)$ be such that $\mu \in \overline{f(Z)} \setminus f (Z)$, where $\overline{f(Z)}$ is the $\tau_1$-closure of $f(Z)$. From \eqref{kell10} $\mu \in f_a (\overline{i_2 (Z)})$, i.e., there exists $\nu \in \overline{i_2 (Z)}$ such that $\mu = \nu|_{(X,\mathcal{M}_2)}$. 

Let $\mathcal{A} \subseteq \mathcal{M}_2$ be a field such that $\# \mathcal{A} < \infty$. Then for any $\varepsilon > 0$: $\{ \mu' \in Z : | \mu' (A) -\nu (A) | < \varepsilon, \ A \in \mathcal{A} \} \neq \emptyset$. However, Lemma \ref{MRN} implies that there exists $\mu^\ast \in \Delta (X,\mathcal{M}_2)$ such that $\nu = \mu^\ast|_{(X,(\mathcal{A} \cup \mathcal{M}_1))}$, where $(\mathcal{A} \cup \mathcal{M}_1)$ is the field generated by $\mathcal{A} \cup \mathcal{M}_1$, i.e., $f^{-1} (\{ \mu \}) \nsubseteq \complement{Z}$\footnote{$\complement{Z}$ is for the complement of set $Z$.}, which is a contradiction.
\end{proof}

Notice that the measurable spaces of the inverse system of probability measure spaces \eqref{hb0} are products of measurable spaces of the inverse systems of probability measure spaces 

\begin{equation}\label{hb3}
((S,\mathcal{A},p_{n+1} (\theta)|_{(S,\mathcal{A})}),\mathbb{N} \cup \{-1\},id_S)
\end{equation}

\noindent and \eqref{hb1}.

Consider the following inverse system 

\begin{equation}\label{hb4}
((S \times \Theta_n,\mathcal{A} \times \mathcal{A}^\ast_n,p_{n+1} (\theta))|_{\mathcal{A} \times \mathcal{A}^\ast_n},\mathbb{N} \cup \{-1\},(id_S,q_{mn})) \ ,  
\end{equation}

\noindent where $\mathcal{A} \times \mathcal{A}^\ast_n$ is for the field generated by the cylindrical sets of $\mathcal{A}$ (the $\sigma$-field on $S$) and $\mathcal{A}^\ast_n$ (the $\sigma$-field on $\Theta_n$).

Moreover, notice that the inverse system \eqref{hb4} can be "embedded" uniquely into the inverse system

\begin{equation}\label{hb5}
((S \times \Theta_n,\mathcal{A} \times B(\Theta_n,\tau_p),\nu_n),\mathbb{N} \cup \{-1\},(id_S,q_{mn})) \ ,  
\end{equation}

\noindent where $\mathcal{A} \times B(\Theta_n,\tau_p)$ is for the field generated by the cylindrical sets of $\mathcal{A}$ and $B(\Theta_n,\tau_p)$, $\nu_n$ is a probability measure such that $\nu_n|_{\mathcal{A} \times \mathcal{A}^\ast_n} = p_{n+1} (\theta)|_{\mathcal{A} \times \mathcal{A}^\ast_n}$ and $\nu_n|_{B(\Theta_n,\tau_p)} = \mu_n$ ($\mu_n$ is from \eqref{hb11}). That is, the inverse system \eqref{hb5} is the "product" of the inverse system of probability measure spaces \eqref{hb3} and \eqref{hb11}.  

\begin{proof}[The proof of Lemma 5.1]
We prove the two player case, see Diagram \ref{limit2a}; the proof of the general case is a slight modification of that of the two player case.

It is clear that it is enough to show that the inverse system \eqref{hb5} admits a unique inverse limit.

Let $n > m$ and $B_j \times C_j$, $j=1,\ldots,k$, where $B_j \in \mathcal{A}$ and $C_j \in B(\Theta_n,\tau_p)$. Then Lemma \ref{fontos4} implies that $C_j$ can be a closed set, therefore from Lemma \ref{fontos3}

\begin{equation*}
(id_S,q_{mn}) (B_j \times C_j) \in \mathcal{A} \times B(\Theta_n,\tau_p) \ , \mspace{10mu} j=1,\ldots,k \ ,
\end{equation*}

\noindent and 

\begin{equation*}
\nu_m \left( \bigcup \limits_{j=1}^k (id_S,q_{mn})(B_j \times C_j) \right) \geq \nu_n \left( \bigcup \limits_{j=1}^k B_j \times C_j \right) \ .
\end{equation*}

Therefore the inverse system \eqref{hb5} is $\varepsilon$-complete, so Theorem \ref{lemmakell1} implies that it admits a unique inverse limit.
\end{proof}


\end{document}